\documentclass[onecolumn,11pt]{IEEEtran}
\usepackage{graphicx,setspace}
\usepackage{epsf}
\usepackage[small,bf]{caption}
\usepackage{amsmath,amsthm,verbatim,amssymb,amsfonts,amscd, graphicx,algorithmic,algorithm}
\usepackage{tikz}
\usepackage{soul}

\usepackage{subfig}
\newcommand{\CC}{\mathbb{C}}
\newcommand{\RR}{\mathbb{R}}
\newcommand{\lm}{\lambda}
\newcommand{\SSS}{\mathcal{S}}
\newcommand{\Gm}{\Gamma}
\newcommand{\rr}{\rho}
\newcommand{\ra}{\rightarrow}
\newcommand{\e}{\epsilon}
\newcommand{\PP}{\mathbb{P}}

\newcommand{\ben}{\begin{eqnarray}}

\newcommand{\een}{\end{eqnarray}}

\newtheorem{lemma}{Lemma}

\newtheorem{theorem}{Theorem}

\title{On Block Coherence of Frames}

\author{%
  \authorblockN{Robert Calderbank\authorrefmark{1}, Andrew Thompson\authorrefmark{1}, Yao Xie\authorrefmark{2}
  }\\
  \authorblockA{\normalsize
      \authorrefmark{1}Department of Mathematics, Duke University, Durham, NC.\\
\authorrefmark{2}H. Milton Stewart School of Industrial and Systems Engineering, Georgia Institute of Technology, Atlanta, GA.
  }

  Email: robert.calderbank@duke.edu, thompson@math.duke.edu, yao.xie@isye.gatech.edu.
}

\begin{document}
\maketitle

\begin{abstract}

Block coherence of matrices plays an important role in analyzing the performance of block compressed sensing recovery algorithms (Bajwa and Mixon, 2012). In this paper, we characterize two block coherence metrics: worst-case and average block coherence. First, we present lower bounds on worst-case block coherence, in both the general case and also when the matrix is constrained to be a union of orthobases. We then present deterministic matrix constructions based upon Kronecker products which obtain these lower bounds. We also characterize the worst-case block coherence of random subspaces. Finally, we present a flipping algorithm that can improve the average block coherence of a matrix, while maintaining the worst-case block coherence of the original matrix. We provide numerical examples which demonstrate that our proposed deterministic matrix construction performs well in block compressed sensing.
	
\end{abstract}

\section{Introduction}\label{intro}
In compressed sensing~\cite{donoho,candes1}, a new paradigm in signal processing, we are interested in recovering a sparse signal $x$ from a reduced number of measurements $y = Ax$. Rather than sparsity, where $x$ only has a few nonzero entries, a more general model is to consider {\em block sparsity}, where the signal has a block structure in that the entries in each block are nonzero simultaneously, and there are only a few blocks with nonzero entries. Recovery of block sparse signal arises from a wide range of applications (as described in Section \ref{sec:applications}).

Various algorithms have been proposed to recover block sparse signals, including \textit{group LASSO}~\cite{eldar} and the lower complexity \textit{one-step group thresholding} \cite{one_step}, see Section \ref{sec:applications} for a brief survey. It has been shown that the performance of one-step group thresholding depends crucially on two block coherence properties of the sensing matrix $A$: {\em worst-case block coherence} and {\em average block coherence}, which we next define. Suppose the matrix $A\in\CC^{n\times mr}$ consists of the concatenation of $m$ equally sized blocks, namely
\begin{align*}
&A = \begin{bmatrix} A_1 & A_2 & \cdots & A_m \end{bmatrix},\\
&A_i\in\CC^{n\times r}, \quad i \in\{1,2,\ldots,m\},
\end{align*}
where we assume that the columns of $A$ are unit-norm, that the columns within each block are orthonormal, and also that $mr>n$ and $r<n$. Thus, provided the columns of $A$ span $\CC^n$, $A$ is  a unit-norm frame in $\CC^n$. At times, we will also consider the special case in which $A$ is real. Following \cite{eldar,conditioning}, we define $\mu(A)$, the worst-case block coherence of $A$, to be
\begin{equation}\label{worstcase_def}
\mu(A):=\max_{i\neq j}\|A_i^{\ast}A_j\|_2,
\end{equation}
and following \cite{one_step}, we define $\nu(A)$, the average block coherence of $A$, to be
\begin{equation}\label{average_def}
\nu(A):=\frac{1}{m-1}\max_{i}\Big\|\sum_{j\neq i}A_i^{\ast}A_j\Big\|_2,
\end{equation}
where $X^*$ denotes the conjugate transpose of a matrix $X$, and $\|X\|_2$ denotes the spectral norm of $X$, which is also equal to the maximum absolute singular value of $X$. Note that $\nu(A)$ is referred to as average \textit{group} coherence in~\cite{one_step}.

Intuitively, a matrix $A$ with good coherence properties should have subspaces formed by submatrices of $A$ as orthogonal to each other as possible. Hence, the problem of designing a matrix with good coherence properties is related to the problem of packing subspaces, also known as {\em Grassmann packing}~\cite{packing,alternating}. Let $G(n,r)$ be the Grassmann manifold of $r$-dimensional subspaces in $\CC^{n}$. Now suppose $\{A_1,A_2,\ldots,A_m\}$ are orthonormal bases for the subspaces $\{\SSS_1,\SSS_2,\ldots,\SSS_m\}$, so that $\SSS_i\in G(n,r)$ for each $i=1,2,\ldots,m$. A Grassmann packing of $m$ subspaces in $G(n,r)$ with respect to some distance metric is said to be optimal if it maximizes the minimum distance between the subspaces. Various distance metrics have been considered, including the chordal distance, spectral distance and geodesic distance~\cite{packing,alternating}.

In a seminal paper on Grassmann packings, Conway et al.~\cite{packing} proved an upper bound on the chordal distance given $(m,n,r)$, from which an upper bound on the spectral distance was deduced in~\cite{alternating}. Meanwhile, a fundamental lower bound on worst-case block coherence was essentially (though not explicitly) determined in~\cite{lemmens_seidel}. It can be shown that optimal Grassmann packings with respect to the spectral distance yield matrices with minimum worst-case block coherence. In fact, it is straightforward to deduce the lower bound on worst-case block coherence in~\cite{lemmens_seidel} from the upper bound on spectral distance in~\cite{alternating} (see Section~\ref{lower_bound} for further elaboration).

Concerning achievability of the above-mentioned bounds, most explicit constructions currently existing in the literature are for the case of $r=1$ (worst-case and average column coherence respectively~\cite{gabor}). In this case, the lower bound on worst-case coherence, known as the Welch bound~\cite{welch}, is obtained if and only if the matrix $A$ is an equiangular tight frame (ETF)~\cite{strohmer}, see Section~\ref{constructions}. Several infinite families of ETFs have been found~\cite{steiner,existence,complex,kirkman}, as well as other infinite families which nearly meet the Welch bound~\cite{revisiting}. For the case $r\geq 2$ which is our interest in the present paper, various optimal Grassmann packings (thereby yielding matrices with optimal worst-case block coherence) were constructed for small dimensions in~\cite{packing}. Infinite families of nearly-optimal Grassmann packings have also been found~\cite{shor_sloane,group_theoretic}, both of which have an underlying group-theoretic structure. Numerical methods for empirically constructing Grassmann packings have also been proposed, based upon nonlinear optimization~\cite{trosset} and alternating projection~\cite{alternating}.

It was shown in \cite{lemmens_seidel} that optimal Grassmann packings can also be constructed as the Kronecker product of an ETF with a unitary matrix, from which it immediately follows that every infinite family of ETFs give rise to infinite families of Grassmann packings. In this paper, we will extend this technique of employing Kronecker products.

The contributions of this paper are as follows
\begin{enumerate}
\item We establish a connection between worst-case coherence and optimal Grassmann packings with respect to both the chordal distance and spectral distance. Using this connection, we review and unite existing results concerning optimal Grassmann packings~\cite{packing,alternating} and a fundamental lower bound on worst-case block coherence~\cite{lemmens_seidel}. We then review an optimal construction, originally given in~\cite{lemmens_seidel}: these matrices are constructed as the Kronecker product of an equiangular tight frame with a unitary matrix. The result implies the existence of several families of matrices with optimal worst-case block coherence, and furthermore the existence of several infinite families of optimal Grassmann packings with respect to both the chordal and spectral distances. We also translate bounds on the maximum number of subspaces in optimal Grassmann packings into the world of frame coherence, which gives upper bounds on the number of blocks in the matrix $A$.
\item We also prove a tighter lower bound on the worst-case block coherence of matrices which are constrained to be unions of orthonormal bases. By extending the Kronecker product approach, we explicitly construct matrices which nearly meet the fundamental lower bound which allow for an increased number of blocks. This result implies the existence of several infinite families of matrices with nearly optimal worst-case block coherence, which in turn gives infinite families of near optimal Grassmann packings. We also obtain upper bounds on the number of blocks in the matrix $A$ for optimal packings in this relaxed sense.
\item We analyze the worst-case block coherence of matrices formed from random orthonormal bases. In a proportional-dimensional asymptotic framework, we prove an upper bound on worst-case block coherence which is a small multiple above optimal. Our result can also be used to quantify the asymptotic spectral distance of Grassmann packings formed from the union of random subspaces.
\item We present a flipping algorithm for subspaces, inspired by a similar flipping algorithm for columns \cite{fundamental}, which can reduce the average block coherence of a frame while maintain the worst-case block coherence. We also prove that, starting with a frame with low worst-case block coherence, the flipping algorithm can produce a frame with average block coherence that meets the requirement of the one-step thresholding algorithm \cite{one_step}.
\end{enumerate}

An outline of the paper is as follows: We first extend our introduction in Section~\ref{sec:applications} by exploring some of the applications of block sparsity in compressed sensing, and of optimal Grassmann packings. We focus on worst-case block coherence in Section~\ref{worst_case}: we review the exisiting results on fundamental bounds and optimal constructions in Section \ref{lower_bound}, before presenting our new bounds and constructions for unions of orthonormal bases in Section~\ref{deterministic}, and analyzing random designs in Section~\ref{random}. We then turn to average block coherence in Section~\ref{average}: we prove results for our deterministic constructions in Section~\ref{theoretical}, and we present our flipping algorithm for reducing average block coherence in Section~\ref{flipping_section}. Numerical illustrations are provided in Section~\ref{numerics}.

\subsection{Applications}\label{sec:applications}

There exist several applications of block sparsity in compressed sensing. A far from exhaustive list includes the blind sensing of multi-band signals~\cite{eldar}, DNA microarrays~\cite{DNA}, and also in medical imaging including ECG imaging~\cite{ECG} and source localization in EEG/MEG brain imaging~\cite{EEG}. The same problem has also received much attention in statistical regression, where one predictor may often imply the presence of several others~\cite{one_step}. Another special case is the multiple measurement vector (MMV) model, also referred to as joint or simultaneous sparsity, where the matrix is a suitable rearrangement of the rows of a matrix of the form $I\otimes A$, and which has found applications in sensor networks and MIMO channels~\cite{MIMO2}.

We recall that block coherence is closely related to optimal Grassmann packing, which has found applications in the design of codes for multiple-antenna communication systems, such as MIMO. In this application, complex Grassmann packings consisting of a small number of low dimensional subspaces in a much larger ambient dimension are required, see~\cite{alternating} for references therein.

\subsection{Algorithms for block compressed sensing: further background}

Block coherence is crucial for the analysis of the one-step group thresholding algorithm for block sparse compressed sensing. To give additional context, several algorithms have been proposed for this problem, including the group LASSO convex relaxation (also known as $\ell_2/\ell_1$ mixed optimization)~\cite{eldar}, the greedy block OMP algorithm~\cite{eldar} and iterative thresholding~\cite{AMP}. It is celebrated that compressed sensing allows the recovery of a signal whose sparsity is proportional to its dimension, and several optimal-order guarantees of this form have been obtained in the case of random matrices for algorithms for block sparsity by means of the block-based Restricted Isometry Property and the analysis of message passing algorithms~\cite{eldar,AMP}. The only current optimal-order recovery guarantee for deterministic matrices is for one-step group thresholding~\cite{one_step}, and relies crucially upon both the worst-case and average block coherence. Another important consideration is that one-step group thresholding is an extremely simple algorithm and has the lowest computational complexity among all the algorithms mentioned.

\section{Worst-case block coherence}\label{worst_case}

\subsection{Fundamental bounds and constructions for worst-case block coherence}\label{lower_bound}

In this section, we present a review of existing results concerning optimal Grassmann packings~\cite{packing,alternating} and optimal worst-case block coherence~\cite{lemmens_seidel}, emphasizing the connection between the two. We also present optimal constructions which generalize the one given in~\cite{lemmens_seidel}.

\vspace{0.2cm}

\subsubsection{Bounds}

We begin by making the observation that $\mu(A)=1$ whenever $n<2r$, since in this case the two subspaces corresponding to any pair of blocks have a non-trivial intersection. We will therefore assume throughout Section~\ref{worst_case} that $n\geq 2r$.

Two choices of distance metric for Grassmann packings which have been previously studied~\cite{packing,alternating} are the chordal distance $d_C(\SSS_i,\SSS_j)$ and the spectral distance $d_S(\SSS_i,\SSS_j)$, defined as follows~\cite{alternating}.
\begin{equation}\label{chordal_def}
d_C(\SSS_i,\SSS_j):=\sqrt{r-\|A_i^{\ast}A_j\|_F^2};
\end{equation}
\begin{equation}\label{spectral_def}
d_S(\SSS_i,\SSS_j):=\sqrt{1-\|A_i^{\ast}A_j\|_2^2}.
\end{equation}
where  $\|X\|_F$ denotes the Frobenius norm of $X$, defined as $\|X\|_F = \sum_{i}\sum_j |X_{ij}|^2$.

Both metrics may also be expressed in terms of the principal angles between $\SSS_i$ and $\SSS_j$~\cite{alternating}: the spectral norm is the sine of the smallest principal angle and the square of the chordal distance is the sum of the squared sines of all the principal angles.

The following upper bound on the minimum chordal distance was obtained in a seminal paper by Conway et al.~\cite{packing}. It was named the Rankin bound since it was obtained by projecting Grassmannian space onto a sphere and applying the Rankin bound for spherical codes.

\begin{lemma}[Rankin bound for chordal distance{~\cite[Corollary 5.2]{packing}}]\label{chordal_lemma}
$$\min_{i\neq j} [d_C(\SSS_i,\SSS_j)]^2\le\frac{r(n-r)}{n}\cdot\frac{m}{m-1}.$$
If the bound is met, all subspaces are equidistant.
\end{lemma}

The subspaces $\{\SSS_1,\SSS_2,\ldots,\SSS_m\}$ are said to be \textit{equi-isoclinic} if the singular values among all $A_i^{\ast}A_j$ for $i\neq j$ are equal~\cite{lemmens_seidel}, which is also equivalent to saying that the principal angles are all equal for each pair of subspaces. Dhillon et al.~\cite{alternating} deduced from Lemma~\ref{chordal_lemma} an upper bound on the minimum spectral distance.

\begin{lemma}[Spectral distance bound~{\cite[Corollary 4.2]{alternating}}]\label{spectral_lemma}
\begin{equation}\label{spectral_bound}
\min_{i\neq j} [d_S(\SSS_i,\SSS_j)]^2\le\frac{n-r}{n}\cdot\frac{m}{m-1}.
\end{equation}
If the bound is met, the subspaces are equi-isoclinic.
\end{lemma}

A lower bound on worst-case block coherence may be derived from Lemma~\ref{spectral_lemma}.

\begin{theorem}[Universal lower bound on worst-case coherence~{\cite[Theorem 3.6]{lemmens_seidel}}]\label{block_welch}
\begin{equation}\label{welch_block}
\mu(A)\geq\sqrt{\frac{mr-n}{n(m-1)}}.
\end{equation}
If the bound is met, all singular values among all matrices $A_i^{\ast}A_j$ for $i\neq j$ are equal in modulus.
\end{theorem}

\begin{proof}
Combining (\ref{spectral_def}) and (\ref{spectral_bound}), we have
\begin{equation}\label{chordal_subst}
1-\max_{i\neq j}\|A_i^{\ast}A_j\|_2^2\le\frac{n-r}{n}\cdot\frac{m}{m-1},
\end{equation}
which rearranges to give
\begin{equation}\label{frob_bound}
\max_{i\neq j}\|A_i^{\ast}A_j\|_2^2\geq 1-\frac{n-r}{n}\cdot\frac{m}{m-1}=\frac{n(m-1)-m(n-r)}{n(m-1)}=\frac{mr-n}{n(m-1)},
\end{equation}
and (\ref{welch_block}) follows. Finally, if (\ref{spectral_bound}) holds with equality, it follows from Lemma~\ref{spectral_lemma} that the subspaces are equi-isoclinic, which implies that all singular values among all matrices $A_i^{\ast}A_j$ for $i\neq j$ are equal in modulus.
\end{proof}

Theorem~\ref{block_welch} was essentially proved in~\cite{lemmens_seidel}, which gives an upper bound on the number of subspaces $m$ for a given principal angle. Theorem~\ref{block_welch} extends the Welch bound for coherence of columns of a matrix~\cite{welch}, and reduces to the Welch bound for $r=1$. Similarly, the condition that the subspaces are equi-isoclinic is an extension of the condition for equality in the Welch bound, namely that the columns of the matrix are required to be an equiangular tight frame~\cite{strohmer}.

Frames with optimal block coherence cannot have arbitrarily many blocks. We next give an upper bound on $m$, the number of blocks in the matrix, or equivalently the number of subspaces in an optimal Grassmann packing. Since we require the subspaces to be equi-isoclinic, we may use the following bound on the maximum number of equi-isoclinic subspaces, along with a tighter bound in the case where $A$ is real.

\begin{theorem}[Subspace bound for exact optimality{~\cite[Theorem 3.5]{lemmens_seidel},~\cite[Theorem 8]{alternating}}]\label{equi_bound}
The number of $r$-dimensional equi-isoclinic subspaces in $\CC^n$ cannot exceed $n^2-r^2+1$. The number of $r$-dimensional equi-isoclinic subspaces in $\RR^n$ cannot exceed $\frac{1}{2} n(n+1)-\frac{1}{2} r(r+1)+1$.
\end{theorem}
\vspace{0.2cm}
\subsubsection{Constructions}\label{constructions}

We turn now to constructions which achieve the lower bound (\ref{welch_block}). Both our constructions in Section~\ref{worst_case} will take the form of a Kronecker product between a matrix with good column coherence properties (block coherence with $r=1$) and a unitary matrix. We begin with a definition. First a definition: a matrix $A$ is said to be an \textit{equiangular tight frame} (ETF) if it satisfies the following conditions:
\begin{itemize}
\item The columns of $A$ are unit norm.
\item The inner products between pairs of different columns are equal in modulus.
\item The columns form a tight frame, that is $AA^{\ast}=(mr/n)I$.
\end{itemize}
ETFs have the important property that they minimize the worst-case column coherence (block coherence with $r=1$), achieving the Welch bound.

\begin{lemma}[Welch bound equality~\cite{strohmer}]\label{welch_lemma}
Let $P\in\CC^{n\times m}$ be an ETF and let $r=1$. Then
$$\mu(P)=\sqrt{\frac{m-n}{n(m-1)}}.$$
\end{lemma}

The next result gives a construction with minimum worst-case block coherence. This result appear in various forms throughout~\cite{lemmens_seidel}, but we include a proof of it here for completeness.

\begin{theorem}[Kronecker product construction 1{~\cite{lemmens_seidel}}]\label{kronecker_ETF}
Let $A=P\otimes Q$ where $P\in\CC^{(n/r)\times m}$ is an ETF and $Q\in\CC^{r\times r}$ is a unitary matrix. Then the columns in each block are orthonormal, and
\begin{equation}\label{kron_bound}
\mu(A)=\sqrt{\frac{mr-n}{n(m-1)}}.
\end{equation}
\end{theorem}

\begin{proof}
Let $A$ be defined as in the theorem. Let us write $P=\begin{bmatrix} p_1 & p_2 &\ldots & p_m\end{bmatrix}$ for the columns of $P$, and
$$p_i=\left[\begin{array}{l}
P_{1i}\\
P_{2i}\\
\vdots
\end{array}\right]$$
for the entries in each column. Then we have
\begin{equation}\label{matrix_mult}
\|A_i^{\ast}A_j\|_2=\Big\|\sum_{s=1}^{n/r}(P_{si} Q)^{\ast}(P_{sj} Q)\Big\|_2=\|\langle p_i,p_j\rangle Q^{\ast}Q\|_2=\langle p_i,p_j\rangle,
\end{equation}
where the last equality holds since $Q$ is unitary. Since $P$ is an ETF, it follows from Lemma~\ref{welch_lemma} that, for $i\neq j$,
$$\max_{i\neq j}|\langle p_i,p_j\rangle|^2 = \frac{m-\frac{n}{r}}{\frac{n}{r}(m-1)}=\frac{mr-n}{n(m-1)},$$
which combines with (\ref{matrix_mult}) to yield
$$\mu(A)=\max_{i\neq j}\|A_i^{\ast}A_j\|_2=\sqrt{\frac{mr-n}{n(m-1)}},$$
which proves (\ref{kron_bound}). It remains to show that the columns in each block are orthonormal. Writing $Q=\begin{bmatrix} q_1 & q_2 &\ldots &q_r\end{bmatrix}$ for the columns of $Q$, then the columns of $A$ are $p_i\otimes q_j$ for $1\le i\le(n/r)$ and $1\le j\le r$, and a standard Kronecker product identity gives
$$(p_i\otimes q_j)^{\ast}(p_i\otimes q_j)=(p_i^{\ast}\otimes p_i)(q_j^{\ast}\otimes q_j)=1,$$
which proves that the columns are unit norm. Finally,
$$A_i^{\ast}A_i=\sum_{s=1}^{n/r}(P_{si} Q)^{\ast}(P_{si} Q)=\langle p_i,p_i\rangle Q^{\ast}Q=I,$$
where the last equality follows since $P$ has unit norm columns and $Q$ is unitary, which implies that the columns within each block are orthonormal, and the proof is complete.
\end{proof}

Kronecker product constructions have also previously been considered in~\cite{eldar}.

To construct a matrix for a given triple $(m,n,r)$ in this way, there must exist an ETF of size $(n/r)\times m$, and therefore $r$ must necessarily divide $n$. ETFs can only exist, however, for certain sizes, see~\cite{existence,complex} for insightful studies. Various ETFs with small dimensions have been explicitly found~\cite{existence,complex}, and several infinite families of ETFs are given in~\cite{steiner,existence,complex,kirkman}. Alternating projection algorithms were also used in~\cite{ETF_alternating} to construct approximations to ETFs.

In our construction, the ETF $P$ must obey $m\le(n/r)^2$~\cite{MIMO}, and division by $r^2$ means that the number of blocks is sub-optimal compared with the bound established in Theorem~\ref{equi_bound}. This sub-optimality is clearly more marked as $r$ increases. If the Kronecker product construction is used in conjunction with any of the infinite families of ETFs mentioned in the previous paragraph, further sub-optimality results from the fact that, for these families, the number of columns in the ETF is far short of the optimal number~\cite{MIMO}.

\subsection{Bounds and constructions for nearly optimal worst-case block coherence}\label{deterministic}

We now present new results concerning bounds and constructions for optimal worst-case block coherence for a more restrictive family of matrices, namely where the matrix $A$ consists of a union of orthonormal bases. Such constructions have proved important since they can possess almost optimal coherence properties while allowing for more frame elements than would be possible with equiangular tight frames~\cite{kerdock}. We will see that unions of orthonormal bases will also provide useful constructions with almost optimal block coherence.

\vspace{0.2cm}

\subsubsection{Bounds}

Let us suppose then that $A$ consists of orthonormal bases of size $n\times n$ which are split into $n/r$ blocks of size $n\times r$, for which we necessarily require that $r$ divides $n$. We have the following result.

\begin{theorem}[{Unions of orthonormal bases}]\label{welch_orth}
Let $A$ be a union of orthonormal bases. Then
\begin{equation}\label{orth_bound}
\mu(A)\geq\sqrt{\frac{r}{n}}.
\end{equation}
If the bound is met, $\|A_i^{\ast}A_j\|_2$ is equal for all pairs of blocks $A_i$ and $A_j$ from different orthonormal bases.
\end{theorem}

\begin{proof}
Since $mr>n$, we have at least two orthonormal bases. Suppose $A$ includes the bases $P\in\CC^{n\times n}$ and $Q\in\CC^{n\times n}$, which are divided into the blocks $\{P_1,P_2,\ldots,P_{n/r}\}$ and $\{Q_1,Q_2,\ldots,Q_{n/r}\}$ respectively. Then we have
\begin{equation}\label{gram_compare}
\sum_{i=1}^{n/r}\sum_{j=1}^{n/r}\|P_i^{\ast}Q_j\|_F^2=\|P^{\ast}Q\|_F^2=n,
\end{equation}
where the last equality follows since $P^{\ast}Q$ is unitary. It follows from (\ref{gram_compare}) by a standard norm inequality that
$$\sum_{i=1}^{n/r}\sum_{j=1}^{n/r}\|P_i^{\ast}Q_j\|_2^2\geq\frac{n}{r},$$
from which it follows that
\begin{equation}\label{divide}
\left(\frac{n}{r}\right)^2\max_{i\neq j}\|P_i^{\ast}Q_j\|_2^2\geq\frac{n}{r},
\end{equation}
and dividing by $(n/r)^2$ gives
$\max_{i\neq j}\|P_i^{\ast}Q_j\|_2^2\geq\frac{r}{n}.$
The same argument applies to all pairs of orthonormal bases, which yields (\ref{orth_bound}). Finally note that the bound only holds if (\ref{divide}) holds, which requires $\|P_i^{\ast}P_j\|_2$ to be equal for all $i$ and $j$, and since this holds for all pairs of orthonormal bases, the theorem is proved. \end{proof}

Note that the bound (\ref{orth_bound}) is tighter than (\ref{welch_block}), since
$$\sqrt{\frac{r}{n}}>\sqrt{\frac{mr-n}{n(m-1)}}$$
whenever $r<n$. However, they are asymptotically equivalent for large $m$, that is to say, for fixed $n$ and $r$,
$$\lim_{m\rightarrow\infty}\sqrt{\frac{mr-n}{n(m-1)}}=\sqrt{\frac{r}{n}}.$$

Analogous to Section~\ref{lower_bound}, a bound on the number of subspaces for unions of orthonormal bases can also be derived, by making another link with the world of Grassmannian packings. As well as the Rankin bound on the chordal distance of Lemma~\ref{chordal_lemma}, the authors of~\cite{packing} also proved a tighter bound which allows for a larger number of subspaces, which we state next.

\begin{lemma}[Tighter Rankin bound for chordal distance~{\cite[Corollary 5.3]{packing},~\cite[Corollary 1]{MIMO}}]\label{tighter_chordal_lemma}
For $m>n^2$,
\begin{equation}\label{tighter_chordal}
\min_{i\neq j} [d_C(\SSS_i,\SSS_j)]^2\le\frac{r(n-r)}{n}.
\end{equation}
If the bound is met, $m\le 2(n+1)(n-1)$. Furthermore, if $A$ is real and $m>\frac{1}{2}n(n+1)$, the bound (\ref{tighter_chordal}) also holds, and if the bound is met, $m\le(n-1)(n+2)$.
\end{lemma}

We can deduce from Lemma~\ref{tighter_chordal_lemma} a relaxed bound on the maximum number of blocks.

\begin{theorem}[{Subspace bound for unions of orthonormal bases}]\label{relaxed_bound}
Let $A$ be a union of orthonormal bases, and suppose $\mu(A)=\sqrt{\frac{r}{n}}$. Then $m\le 2(n+1)(n-1)$, and furthermore if $A$ is real, $m\le(n-1)(n+2)$.
\end{theorem}

\begin{proof}[Proof of Theorem \ref{relaxed_bound}]
Suppose $\mu(A)=\sqrt{\frac{r}{n}}$. Then, by (\ref{worstcase_def}),
$$r\left(1-\max_{i\neq j}\|A_i^{\ast}A_j\|_2^2\right)=r\left(1-\frac{r}{n}\right)=\frac{r(n-r)}{n},$$
which combines with (\ref{chordal_def}) and a standard norm inequality to give
$$\min_{i\neq j} [d_C(\SSS_i,\SSS_j)]^2=r-\max_{i\neq j}\|A_i^{\ast}A_j\|_F^2\geq r\left(1-\max_{i\neq j}\|A_i^{\ast}A_j\|_2^2\right)=\frac{r(n-r)}{n}.$$
It follows that the bound in Lemma~\ref{tighter_chordal_lemma} is met, and so we may apply the lemma to deduce both results.
\end{proof}
\vspace{0.2cm}
\subsubsection{Constructions}

In order to obtain matrices with a larger number of blocks, we now consider constructions of the form $P\otimes Q$ in which $P$ is a union of orthobases. We have the following result.

\begin{theorem}[{Kronecker product construction 2}]\label{kronecker_orth}
Let $A=P\otimes Q$ where $P\in\CC^{(n/r)\times m}$ is a union of orthobases such that the any pair of columns $(p_i,p_j)$ of $P$ from different orthobases satisfies
\begin{equation}\label{P_cond}
p_i^{\ast}p_j=\sqrt{\frac{r}{n}},
\end{equation}
and where $Q$ is a unitary matrix. Then $A$ itself is a union of orthobases, the columns in each block are orthonormal, and
\begin{equation}\label{kron_bound2}
\mu(A)=\sqrt{\frac{r}{n}}.
\end{equation}
\end{theorem}

\begin{proof}
Adopting the notation in the proof of Theorem~\ref{kronecker_ETF}, we may follow the argument in the proof of Theorem~\ref{kronecker_ETF} to deduce (\ref{matrix_mult}). It follows from (\ref{P_cond}) that
$$\max_{i\neq j}|\langle p_i,p_j\rangle| = \sqrt{\frac{r}{n}},$$
which combines with (\ref{matrix_mult}) to yield
$$\mu(A)=\max_{i\neq j}\|A_i^{\ast}A_j\|_2=\sqrt{\frac{r}{n}},$$
which yields (\ref{kron_bound2}). Since $P$ has unit norm columns, we may rehearse the argument of Theorem~\ref{kronecker_ETF} to deduce that the blocks of $A$ are orthonormal. It remains to show that $A$ itself is a union of orthobases. In this regard, let us write $P$ as
\[P=
\begin{bmatrix}
P_1 & P_2 &\cdots &P_{mr/n}
\end{bmatrix}
,\]
where $P_i$ is an orthonormal basis for all $i=1,2,\ldots,mr/n$. Then, for all $i\in\{1,2,\ldots,mr/n\}$, $A_i=P_i\otimes Q$, which is unitary, and the theorem is proved.
\end{proof}

We have therefore constructed a matrix which satisfies the worst-case block coherence lower bound (\ref{orth_bound}). To construct a matrix for a given triple $(m,n,r)$ in this way, there must exist a union of orthobases, $P$, of size $(n/r)\times m$, and therefore $r$ must necessarily divide $n$. In addition, we require the columns of $P$ to satisfy (\ref{P_cond}). Several families of matrices have been proposed which satisfy these conditions, including Alltop Gabor frames, discrete chirp frames and Kerdock frames, see~\cite{revisiting}. Each of these frames allow $P$ to have size $(n/r)\times(n/r)^2$, leading to a maximum of $(n/r)^2$ blocks in $A$.

\subsection{Random designs with low worst-case block coherence}\label{random}

We now analyse the block coherence of matrices consisting of random blocks, restricting our attention to the case of $A$ being real.\footnote{The analysis becomes more involved in the complex case.} By ``random'', we mean that each block in $A$ is an orthonormal basis for a subspace drawn from the uniform distribution on the Grassmann manifold $G(n,r)$ defined in Section~\ref{intro}, which is also the distribution with probability measure invariant under orthogonal transformations of $G(n,r)$~\cite{largest_angle}.

We first introduce some definitions of some special functions and distributions. 
Given $p>0$, define the gamma function~\cite[Sections 6.1]{handbook} to be
\begin{equation}\label{gamma_fn}
\Gm(p):=\int_0^{\infty} t^{p-1}e^{-t}\,dt,
\end{equation}
and, given $p>0$ and $r$ a positive integer, define the multivariate gamma function~\cite{largest_angle} to be
\begin{equation}\label{multi_gamma}
\Gm_r(p):=\pi^{\frac{1}{4} r(r-1)}\prod_{j=1}^r\Gm\left(p+\frac{1-j}{2}\right).
\end{equation}

Let $A_i\in\RR^{n\times r}$ and $A_j\in\RR^{n\times r}$ be orthonormal bases for two random subspaces, and let $\lm_1\geq\lm_2\geq\ldots\geq\lm_r$ be the squared singular values of $A_i^{\ast}A_j$. Then, if $n\geq 2r$, $(\lm_1,\lm_2,\ldots,\lm_r)$ are shown in~\cite{largest_angle} to follow the multivariate beta distribution.

\begin{lemma}[{Distribution of the squared singular values of $A_i^{\ast}A_j$~\cite{largest_angle}}]\label{random_pdf}
Let $A_i\in\RR^{n\times r}$ and $A_j\in\RR^{n\times r}$ be orthonormal bases for two real random subspaces, where $n\geq 2r$, and let $\lm_1\geq\lm_2\geq\ldots\geq\lm_r$ be the squared singular values of $A_i^{\ast}A_j$. Then $(\lm_1,\lm_2,\ldots,\lm_r)\sim \mathrm{Beta}_r\left(\frac{r}{2},\frac{n-r}{2}\right)$, with pdf
\begin{equation}\label{beta_pdf}
f(\lm_1,\lm_2,\ldots,\lm_r)=c_{n,r}\prod_{i<j}(\lm_i-\lm_j)\cdot\prod_{i=1}^r\lm_i^{-\frac{1}{2}}(1-\lm_i)^{\frac{1}{2}(n-2r-1)};\;\;\;\;\lm_i\geq 0\;\;\forall\;\;i,
\end{equation}
where
$$c_{n,r}:=\frac{\pi^{\frac{1}{2} r^2}\Gm_r\left(\frac{n}{2}\right)}{\left[\Gm_r\left(\frac{r}{2}\right)\right]^2\Gm_r\left(\frac{n-r}{2}\right)},$$
and where $\Gm_r(\cdot)$ is defined in (\ref{multi_gamma}).
\end{lemma}

We are interested in $\lm_1$, the largest squared singular value. Note that the restriction $n\geq 2r$ is quite natural, since we have $\lm_1=1$ if $n<2r$ since the subspaces have nontrivial intersection.

While (\ref{beta_pdf}) is somewhat opaque, it is possible to derive from it a quantitative \textit{asymptotic} bound on $\lm_1$ in a framework in which the dimensions $(n,r)$ grow proportionally. More precisely, we consider a sequence of matrices $A$ in which the dimensions $(n,r)$ tend to infinity in proportion to each other. We define $\beta\in(0,1]$ to be the limiting value of the ratio $r/n$, that is
$$\beta=\lim_{n\rightarrow\infty}\frac{r}{n}.$$
Given $\beta\in(0,1/2)$, we determine a small constant $\hat{a}(\beta)$ (see later in the section for a quantification) such that the probability that $[\mu(A)]^2$ exceeds $\hat{a}(\beta)\cdot\beta$ vanishes exponentially as $n\rightarrow\infty$. Practically speaking, our asymptotic result will accurately describe the behaviour of random subspaces in large dimensions.

The significance of this result is the way it relates the worst-case block coherence of random subspaces to the fundamental limits derived in Section~\ref{lower_bound}. There it was shown that a lower bound on $[\mu(A)]^2$ is given by $(mr-n)/[n(m-1)]$, which tends to $\beta$ as $(m,n,r)\rightarrow\infty$ proportionally. It follows that, in this asymptotic sense, random subspaces gives worst-case block coherence with optimal order, though the multiple involved is slightly sub-optimal.

We next define our threshold $\hat{a}(\beta)$. Let $\beta\in(0,1/2)$, and let $\hat{a}(\beta)$ be the unique solution in $2\le a<1/\beta$ to the equation
\begin{equation}\label{threshold_def}
\beta\ln a+\left(\frac{1-2\beta}{2}\right)\ln(1-a\beta)-(1-\beta)\ln(1-\beta)=0.
\end{equation}
Equation (\ref{threshold_def}) has a unique solution in $a>2$, since the left-hand side is positive for $a=2$, tends to $-\infty$ as $a\ra 1/\beta$, and is monontonically decreasing on $[2,1/\beta)$. We proceed to our main result for random subspaces, the proof of which can be found in the Appendix. Note that the restriction $\beta\in(0,1/2)$ corresponds to the restriction $n>2r$. The case $n=2r$ is not covered by our analysis, since in that case the proportional-dimensional analysis does not apply (see~\ref{random_proofs}). However it can be shown that $\hat{a}(\beta)\ra 2$ as $\beta\ra 1/2$, and therefore one expects the result $a(1/2)=2$, which would imply that the subspaces have non-trivial intersection with high probability.

\newpage

\begin{theorem}\label{random_result}
Let $r/n\rightarrow\beta\in(0,1/2)$ as $n\rightarrow\infty$ and choose $\e>0$. Let $A$ be a real matrix consisting of random blocks, as defined earlier in this section, where the number of blocks $m$ is polynomial in $n$. Then
$$\PP\left\{[\mu(A)]^2\geq\hat{a}(\beta)\cdot\beta+\e\right\}\ra 0,$$
where $\hat{a}(\beta)$ is defined in (\ref{threshold_def}).
\end{theorem}

During the final revision of this paper, we became aware of some closely related work in~\cite{bodmann}, which also considers the same random subspace model as in the present paper, and proves that the \textit{Hilbert-Schmidt inner product}, defined to be $\mbox{Tr}[(A_i A_i^{\ast})(A_j A_j^{\ast})]=\|A_i^{\ast}A_j\|_F^2$, asymptotically approaches $r^2/n$ for all $i\neq j$. This quantity can be viewed as a counterpart of the worst-case block coherence considered here, where the spectral distance metric is replaced by chordal distance. Since 
\begin{equation}\label{frob_spec}
\|A_i^{\ast}A_j\|_F^2\le r\cdot\|A_i^{\ast}A_j\|_2^2,
\end{equation}
our result can be used to deduce a similar result for the Hilbert-Schmidt inner product. A significant difference between the respective results is that ours is given in a proportional-dimensional asymptotic framework alone, whereas a result is given in~\cite{bodmann} for finite dimensions (but which also allows asymptotic results to be deduced). It is also interesting that the result in~\cite{bodmann} involves no multiplicative constant: the Hilbert-Schmidt inner product approaches exactly $r^2/n$. This points to the fact that, asymptotically, the inequality (\ref{frob_spec}) is not satisfied as an equality for random subspaces, and that the eigenvalues of $A_i^{\ast}A_j$ do not concentrate around their expected value but rather tend to some fixed distribution. This fits entirely with the existence of a known limiting distribution of the multivariate beta distribution~\cite{wachter}. From this viewpoint, our result complements well the one given in~\cite{bodmann}.

Fig. \ref{fig:empirical} plots the square of the asymptotic upper bound on worst-case block coherence of random subspaces, $\hat{a}(\beta)$, for $\beta\in(0,1/2)$. We see that worst-case block coherence is never more than a small multiple\footnote{As $\beta\ra 0$, $\hat{a}(\beta)$ tends to $\approx 5.357$, which is the solution in $2\le a<1/\beta$ to $a=2(1+\ln a)$, and this upper bounds $\hat{a}(\beta)$ on $\beta\in(0,1/2)$.} of $\sqrt{\beta}\approx\sqrt{r/n}$. Also plotted is the empirical behavior of worst-case coherence for random subspaces: fixing $n=1000$, for varying $r$, $1000$ trials of $(n/r)^2$ random subspaces were generated (mimicking the Kronecker product constructions of Section~\ref{deterministic}), and the worst-case coherence calculated as an average over all trials. We observe that the theoretical upper bound is somewhat tight, and that the expected asymptotic behavior is in evidence even for this modest problem size of $n=1000$. Note that empirical tests take $\beta\geq 0.05$, since large-scale testing is required to obtain high-dimensional limiting behavior for smaller values of $\beta$.

\begin{figure}[h]
\begin{center}
\includegraphics[width = 0.4\textwidth]{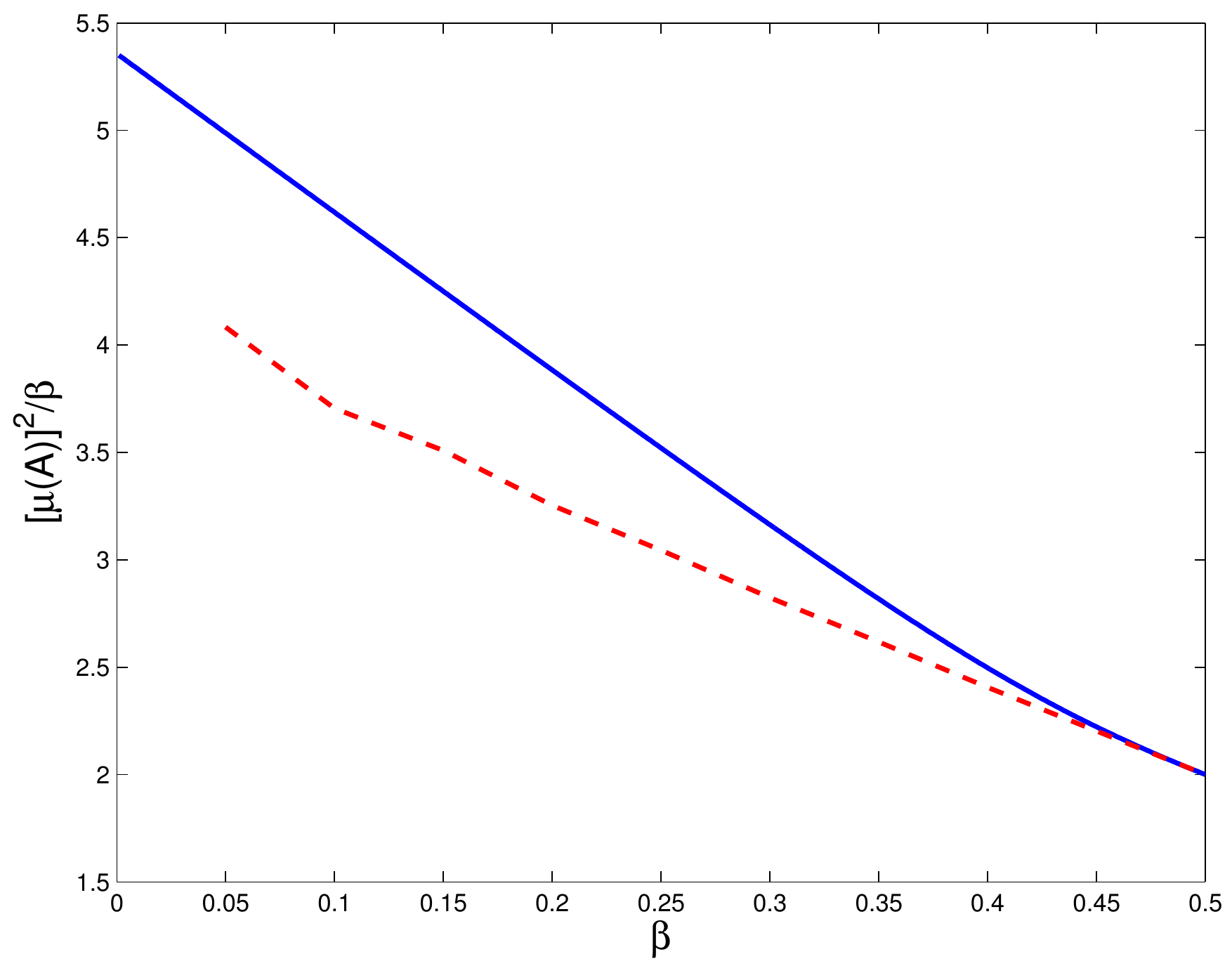}\caption{Worst-case block coherence for random subspaces: theoretical upper bound (unbroken) and the empirical average from $1000$ trials with $n=1000$ and $m=(n/r)^2$ (dashed).}
\label{fig:empirical}
\end{center}
\end{figure}

\section{Average coherence}\label{average}

\subsection{Theoretical results}\label{theoretical}

The constructions which consist of unions of orthobases are also amenable to the analysis of their average block coherence, defined in (\ref{average_def}). Given a construction $A=P\otimes Q$, where $Q$ is unitary, we relate the average block coherence of $A$ to the average column coherence of $P$~\cite{revisiting}, defined to be
\begin{equation}\label{average_col_def}
\nu_1(P):=\frac{1}{m-1}\max_{i}\Big\|\sum_{j\neq i}p_i^{\ast}p_j\Big\|_2,
\end{equation}
where $\{p_i\}$ are the columns of $P$.

\begin{lemma}\label{kron_lemma}
Let $A=P\otimes Q$ where $P\in\CC^{(n/r)\times m}$, and where $Q\in\CC^{r\times r}$ is a unitary matrix. Then
\begin{equation}\label{average_link}
\nu(A)=\nu_1(P).
\end{equation}
\end{lemma}

\begin{proof}
Adopting the notation in the proof of Theorem~\ref{kronecker_ETF}, we may follow the argument in the proof of Theorem~\ref{kronecker_ETF} to deduce (\ref{matrix_mult}). Then, using (\ref{average_def}) and (\ref{average_col_def}), we have
$$\nu(A)=\frac{1}{m-1}\max_{i}\Big\|\sum_{j\neq i}A_i^{\ast}A_j\Big\|_2 =\frac{1}{m-1}\max_{i}\Big|\sum_{j\neq i}p_i^{\ast}p_j\Big|=\nu_1(P),$$
which proves (\ref{average_link}).
\end{proof}

Using results from~\cite{revisiting} in combination with Lemma~\ref{kron_lemma}, we obtain specific results for $P$ belonging to several families of frames, which are displayed in Table~\ref{coherence_table}. See~\cite{revisiting} for more details on how these frames are constructed.

\begin{table}
\caption{Worst-case and average block coherence for $A=P\otimes Q$, for $Q$ unitary and $P$ belonging to various families of frames.}
\centering
\begin{tabular}{|c||c|c|}
\hline Frame family&$\mu(A)$&$\nu(A)$\\ \hline \hline
Alltop Gabor&$\frac{1}{\sqrt{n}}$&$\frac{1}{n+1}$\\ \hline
Discrete chirp&$\frac{1}{\sqrt{n}}$&$\frac{m-n}{n(m-1)}$\\ \hline
Kerdock&$\frac{1}{\sqrt{n}}$&$\frac{1}{m-1}$\\ \hline
Dual BCH&$\sqrt{\frac{2}{n}}$&$\frac{m-n}{n(m-1)}$\\ \hline
\end{tabular}
\label{coherence_table}
\end{table}


\subsection{Flipping algorithm}\label{flipping_section}

In the following we will present an algorithm that can improve the average coherence by random rotations, while preserving the worst-case coherence, since $\|X\|_2 = \| -X \|_2$ for a matrix $X$. This algorithm is inspired by a similar flipping algorithm for columns \cite{fundamental}.

\begin{theorem}\label{existence}
Let $\{A_i\}_{i=1}^m$ be a set of subspaces with $A_{i}\in \mathbb{R}^{n\times r}$, where $a_i = \{1, -1\}$ are equally likely and mutually independent. If $c > 16\sqrt{e}\sqrt{\frac{n}{r(m-1)}}$, then there exists a set of flipped subspaces $\{a_i A_i\}$ with the same worst-case coherence and satisfies $\nu \leq c \mu \sqrt{\frac{r\log m}{n}}$.
\end{theorem}

\begin{proof}
Note that the average coherence of the wiggled subspaces is given by
\begin{align*}
\nu&= \frac{1}{m-1}\max_{i} \Big\|\sum_{j\neq i}  a_i a_j A_i^*  A_j\Big\|_2.
\end{align*}
Now consider for fixed $i = 1$,
\[
\Big\|\sum_{j\neq i} a_i a_j A_i ^* A_j  \Big\|_2
= \Big\|\sum_{j=2}^m a_i a_j A_i^* A_j \Big\|_2
\]
We can construct a sequence $M_k = \sum_{j=1}^{k+1} a_i a_j A_i^* A_j $, and let $M_0 = 0$. The sequence $M_k$, $k = 0, 1, \ldots, m-1$ forms a martingale, since for $k = 1, \ldots, m-1$
\begin{align*}
&&\mathbb{E} \{M_{k} - M_{k-1} \} =
\mathbb{E} \{a_1 a_{k+1} A_1^* A_{k+1}\} \\
&&=
\mathbb{E} \{a_1\} \mathbb{E} \{a_{k+1}\} A_1^* A_{k+1}  = 0,
\end{align*}
where we have used the fact that $a_i$'s are mutually independent.
The martingale difference is  bounded
\begin{align*}
&\|M_{k} - M_{k-1}\|_2 =
\|a_1 a_{k+1}  A_1^* A_{k+1}\|_2 \\
&= \|A_1^* A_{k+1}\|_2 \leq \mu, \quad \forall k = 0, \ldots, m-1.
\end{align*}

Using the Banach-space-valued Azuma's inequality, we obtain the following
\[
\mathbb{P} \left\{
\|M_k - M_0\|_2 \geq \delta \right\}
\leq e^2 \exp\Big\{-\frac{c_0 \delta^2}{(m-1)\mu^2}\Big\}
\]
where $c_0 \triangleq \frac{e^{-1}}{256}$.
Note that
\[
\frac{1}{m-1}\|M_k - M_0\|_2 = \frac{1}{m-1}\Big\|
\sum_{j=2}^m A_1^* A_j\Big\|_2.
\]
Hence taking a union bound with respect to $i$:
\begin{align*}
&\mathbb{P} \Big\{
\frac{1}{m-1}\max_{i}\Big\|\sum_{j\neq i}  a_i a_j A_i^*  A_j\Big\|_2
\geq \nu \Big\}\\
&= \mathbb{P} \Big\{
\max_{i}\Big\|\sum_{j\neq i}  A_i^* A_j\Big\|_2
\geq (m-1)\nu \Big\}\\
&\leq e^2 m \exp\Big\{-\frac{c_0 \nu^2(m-1)}{\mu^2}\Big\}
\end{align*}

Set
\[
\nu \leq c \mu \sqrt{\frac{r\log m}{n}},
\]
then
\begin{align}
&\mathbb{P} \Big\{
\frac{1}{m-1}\max_{i}\Big\|\sum_{j\neq i}  a_i a_j A_i^*  A_j\Big\|_2
\geq \nu \Big\}\nonumber\\
&\leq   \exp\{ 2+
(1 - c_0 c^2 (m-1) \frac{r}{n}) \log m
\} \label{align}
\end{align}
where $c_0 = e^{-1}/256$. Hence
when $c > 16\sqrt{e}\sqrt{\frac{n}{r(m-1)}}$ (note that $n\ll r(m-1)$), $1 - c_0 c^2 (m-1) \frac{r}{n} < 0$. When $\frac{mr}{n}$ is sufficiently large and $m$ is sufficiently large, the right hand side of (\ref{align}) is less than 1.
Under such conditions,
\[
\mathbb{P} \Big\{
\frac{1}{m-1}\max_{i}\Big\|\sum_{j\neq i}  a_i a_j A_i^*  A_j\Big\|_2
\leq \nu
\Big\} > 0,
\]
which proves the theorem.
\end{proof}


\begin{algorithm}
\caption{Flipping algorithm}
\begin{algorithmic}[1]
\STATE Input: $A = \begin{bmatrix} A_1 & \cdots & A_m \end{bmatrix}$
\STATE Output: $B = \begin{bmatrix} B_1 & \cdots & B_m \end{bmatrix}$
\STATE $B_1 = A_1$, $F_1 = B_1$
\FOR{$k=1 \to m-1$}
\IF{$\|F_k + A_{k+1}\|_2 \leq \|F_k - A_{k+1}\|_2$}
\STATE $B_{k+1} = A_{k+1}$
\ELSE
\STATE $B_{k+1} = -A_{k+1}$
\ENDIF
\STATE $F_{k+1} = F_k + B_{k+1}$
\ENDFOR
\end{algorithmic}\label{flipping}
\end{algorithm}

In \cite{one_step}, it is established that for the one-step group thresholding algorithm to perform well, the average block coherence $\nu$ has to be sufficiently small: $\nu \leq c \mu \sqrt{\frac{r\log m}{n}}$. In the following theorem, we demonstrate that the flipping algorithm can produce a set of subspaces $\{A_i\}_{i=1}^m$ that satisfies the condition on block coherence.
\begin{theorem}\label{thm2}
Let $\{A_i\}_{i=1}^m$, $A_i\in \mathbb{R}^{n\times r}$ be a set of subspaces. Suppose \[
\frac{m-1}{m-n/r} \cdot\frac{1}{\log m}\cdot \frac{\sqrt{m}+1}{m-1}
\leq c^2 (\frac{r}{n})^2\]
for some constant $c>0$,
then the flipping algorithm in Algorithm \ref{flipping} outputs a set of subspaces $\{B_i\}_{i=1}^m$ with the same worst-case block coherence $\mu$ and the average block coherence satisfies $\nu \leq c \mu \sqrt{\frac{r\log m}{n}}$.
\end{theorem}

\begin{proof}
Lemma \ref{block_welch} sets a lower-bound on the worse-case block coherence. Combine this with Lemma \ref{lemma2}, which says that the flipping algorithm produces a set of subspaces with $\nu \leq \frac{\sqrt{m}+1}{m-1}$. Hence, to have $\nu \leq c \mu \sqrt{\frac{r\log m}{n}}$ for $\{B_i\}$, it suffices to have
$$\frac{\sqrt{m}+1}{m-1} \leq c \sqrt{\frac{mr-n}{n(m-1)}}\cdot \sqrt{\frac{r\log m}{n}},$$
or equivalently
\[
\frac{m-1}{m-n/r} \cdot\frac{1}{\log m}\cdot \frac{\sqrt{m}+1}{m-1}
\leq c^2 (\frac{r}{n})^2,
\]
for the constant $c$.
\end{proof}

%

The following lemma finds an upper bound on the average block coherence for  the set of subspaces produced by the flipping algorithm:
\begin{lemma}\label{lemma2}
The flipping algorithm produces $\{B_k\}$ whose average coherence satisfies $\nu \leq \frac{\sqrt{m}+1}{m-1}$.
\end{lemma}

\begin{proof}
Note that
\begin{align*}
\Big\|\sum_{j\neq i} A_i^* A_j \Big\|_2&
= \Big\|\sum_{j=1}^m A_i^* A_j - A_i^* A_i \Big\|_2 \\
&\leq \Big\|\sum_{j=1}^m A_i^* A_j\Big\|_2 + \|A_i^* A_i\|_2  = \Big\|\sum_{j=1}^m A_i^* A_j\Big\|_2 + 1 \\
& \leq \|A_i\|_2 \Big\|\sum_{j=1}^m A_j\Big\|_2 + 1 =\Big\|\sum_{j=1}^m A_j\Big\|_2 + 1.
\end{align*}
where we have used the triangle inequality of the matrix norm $\|A + B\|_2 \leq \|A\|_2+ \|B\|_2$, the inequality $\|AB\|_2\leq \|A\|_2\|B\|_2$, $A_i^* A_j = I_r$, and $\|A_i\|_2 = 1$.

Next, we prove the result by induction. First, $\|F_1\|_2 = \|A_1\|_2 = 1$. Next assume that $\|F_k\|_2 = k$. Use the parallelogram law
$
2\|X\|_2^2 + 2\|Y\|_2^2 =\|X+Y\|_2^2 + \|X-Y\|_2^2
$
for matrices $X$ and $Y$, we have that
\begin{align*}
\|F_k\|_2^2 =& - \|A_{k+1}\|_2^2 + \frac{1}{2}(\|F_k + A_{k+1}\|_2^2 + \|F_k - A_{k+1}\|_2^2) \\
\geq & - \|A_{k+1}\|_2^2 +\frac{1}{2} \cdot 2 \|F_k + B_{k+1}\|_2^2 \\
=& -1 + \|F_{k+1}\|_2^2.
\end{align*}
Hence $\|F_{k+1}\|_2^2\leq 1 + \|F_k\|_2^2 \leq k+1$, by induction. Hence $\|F_m\|_2^2\leq m$. By definition of $F_m$, this means that $\frac{1}{m-1}\|\sum_{i=1}^m B_i\|_2\leq \frac{\sqrt{m}+1}{m-1}$.
\end{proof}

\section{Numerical examples}\label{numerics}

In this section, we present numerical examples, to compare random frames formed by random subspaces, with the deterministic frames formed as a Kronecker product of a matrix with an orthogonal matrix.

In Fig. \ref{performance_comp}, we compare the performance of the one-step block thresholding algorithm for support recovery, comparing deterministic frames with those formed from the union of random subspaces. The performance metric is the non-discovery proportion (NDP), which is the fraction of entries in the original support of the signal that are not present in the recovered support. In Fig. \ref{performance_comp}a, the deterministic frame is constructed as  $A = P \otimes Q$, where $P$ is an equiangular tight frame \cite{complex}, and $Q$ is a Hadamard matrix. The size of the frame is $n = 12$, the number of blocks $m = 16$, and the size of the blocks is $r = 2$. The random subspaces are formed by first generating a Gaussian matrix $A \in\mathbb{R}^{n\times mr}$, performing eigendecomposition for each block $A_i \in \mathbb{R}^{n\times 2}$ to find the left singular vectors $U_i$, and finally forming the block matrix $\begin{bmatrix} U_1 & \cdots & U_m\end{bmatrix}$. The different lines in the same block correspond to different dynamic ranges for amplitudes of non-zero elements. Note that the deterministic construction has much better performance than that of the union of random subspaces. Also note the NDP does not depend on dynamic ranges, which is consistent with the result in \cite{one_step}.

In Fig. \ref{performance_comp}b, we consider larger frames with $n = 128$, $m = 2048$, and $r = 2$.  We compare a deterministic construction with $A = P \otimes Q$, where $P$ is a real Kerdock matrix, and $Q$ is a Hadamard matrix. Using a Kerdock matrix for $P$ allows us to pack more blocks in forming a frame, while still obtaining a low coherence property. The random subspaces are constructed from random Gaussian matrices as above. Similarly, the deterministic construction has better performance than the random subspaces, and the performance does not depend on the dynamic range of the non-zero elements. Fig. \ref{Fig:d} demonstrates $\|A_i^*A_j\|_2$ for $i, j = 1, 2, \ldots, 128$ for the deterministic frame and the random subspaces, as well as their worst-case block coherence $\mu$ and average block coherence $\nu$. Note that the deterministic frame has much lower coherence properties and the magnitudes of $\|A_i^* A_j\|_2$ have a structure: the magnitude of $\|A_i^* A_j\|_2$, $i\neq j$ takes only 2 values $\{0,\sqrt{r/n}\}$.

\begin{figure}[h]
\begin{center}
\subfloat[ETF]{\includegraphics[width = 0.4\textwidth]{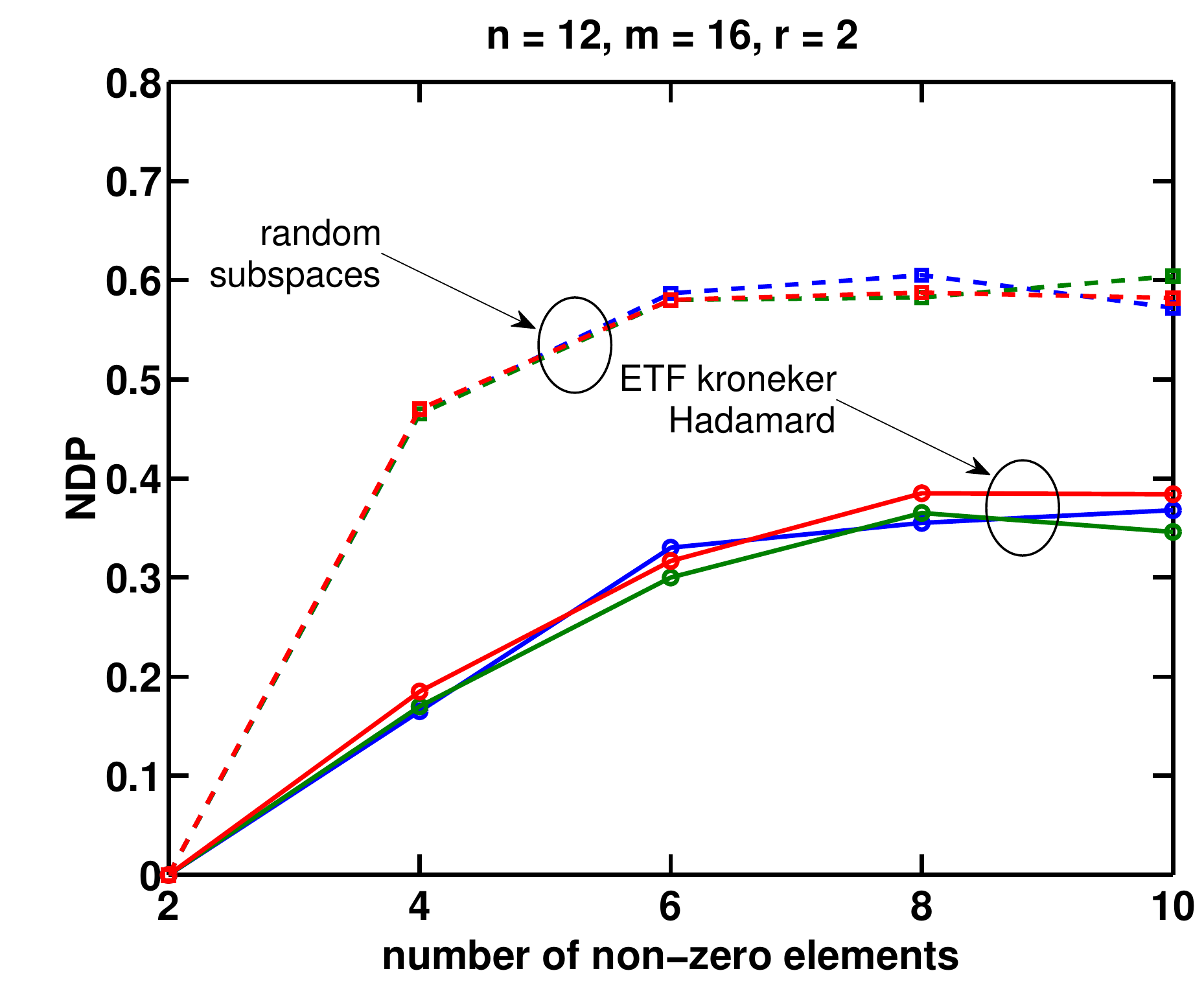}}
\subfloat[Kerdock]{\includegraphics[width = 0.4\textwidth]{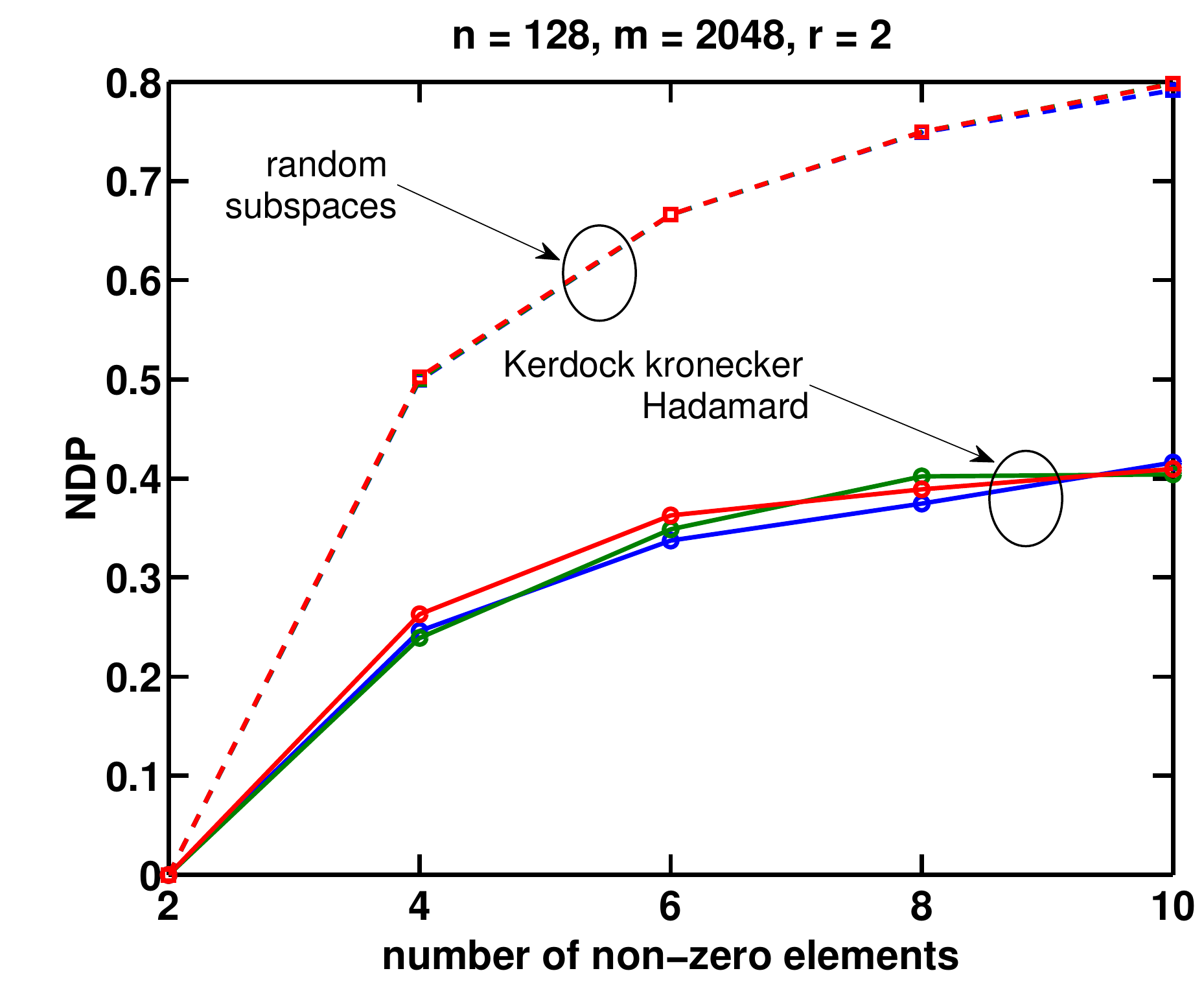}}
\end{center}
\caption{Non-discovery proportion (NDP) for the random subspaces and (a) when $A$ is formed as equiangular tight frame (ETF) with a Hadamard matrix;
(b) when $A$ is formed as the Kronecker product of a Kerdock matrix with a Hadamard matrix.}
\label{performance_comp}
\end{figure}

\begin{figure}[h]
\begin{center}
\subfloat[$\mu = 0.500, \nu = 0.019$]{\label{Fig:d1}
\includegraphics[width = .44\linewidth]{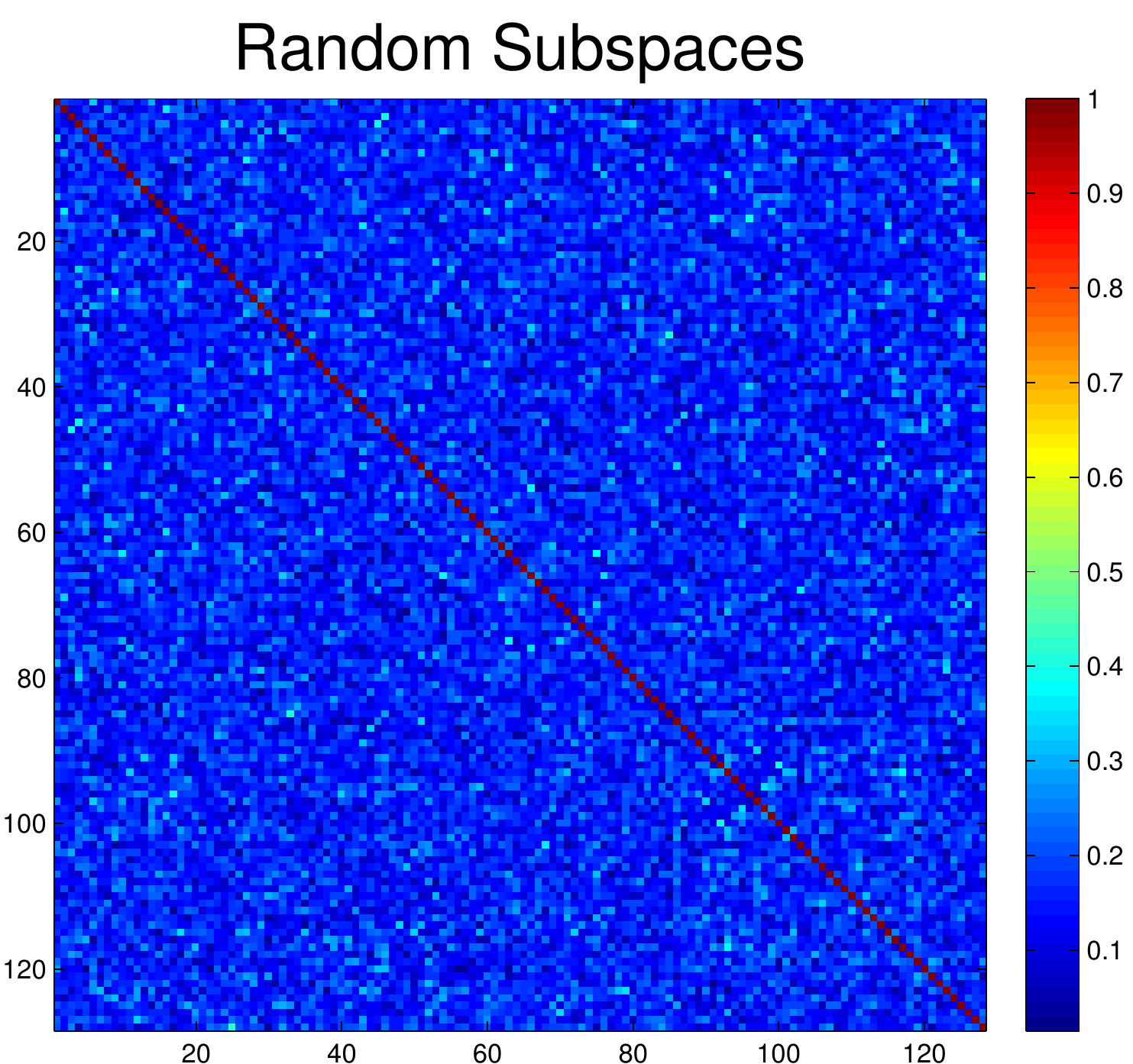}}
\qquad
\subfloat[$\mu = 0.125,\nu = 4.885\times 10^{-4}$]
{\label{Fig:d2}
\includegraphics[width = .44\linewidth]{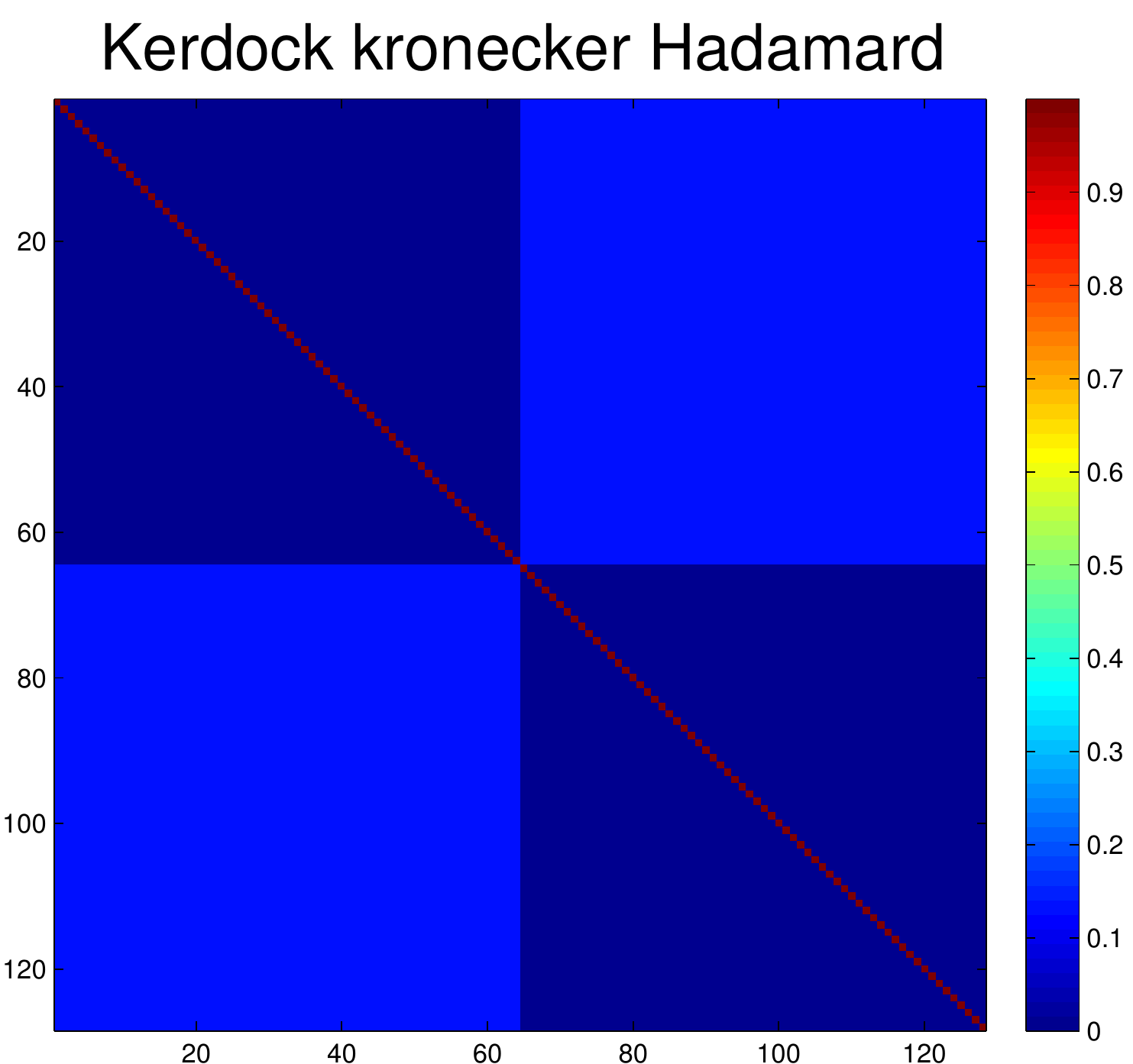}}
\end{center}
\caption{Maps $\|A_i^*A_j\|_2$ for random subspaces, and $A$ formed by Kerdock matrix kronecker with Hadamard matrix.}
\label{Fig:d}
\end{figure}

In the next example, we demonstrate the effectiveness of the flipping algorithm in reducing the average coherence of a frame consisting of random subspaces. We apply the flipping algorithm on 10 realization of random frames with size $n = 128$, $m = 2048$, and $r = 1, 2, 3$, respectively. The flipping algorithm reduces the average coherence of the random frame by a significant percentage.

\begin{table}[h]
\caption{$\nu$ before and after flipping algorithm for union of random subspaces}
\begin{center}
\begin{tabular}{|c||c|c|c|}
\hline
& $r = 1$ & $r = 2$ &  $r = 3$ \\\hline\hline
before &0.0096  &  0.0181  &  0.0111 \\\hline
after & 0.0015  &  0.0027   & 0.0037 \\\hline
$\%$ of improvement &  84.6$\%$  &  85.2$\%$  &  66.6$\%$ \\\hline
\end{tabular}
\end{center}
\end{table}

\begin{figure}[h]
\begin{center}
\includegraphics[width = 0.4\textwidth]{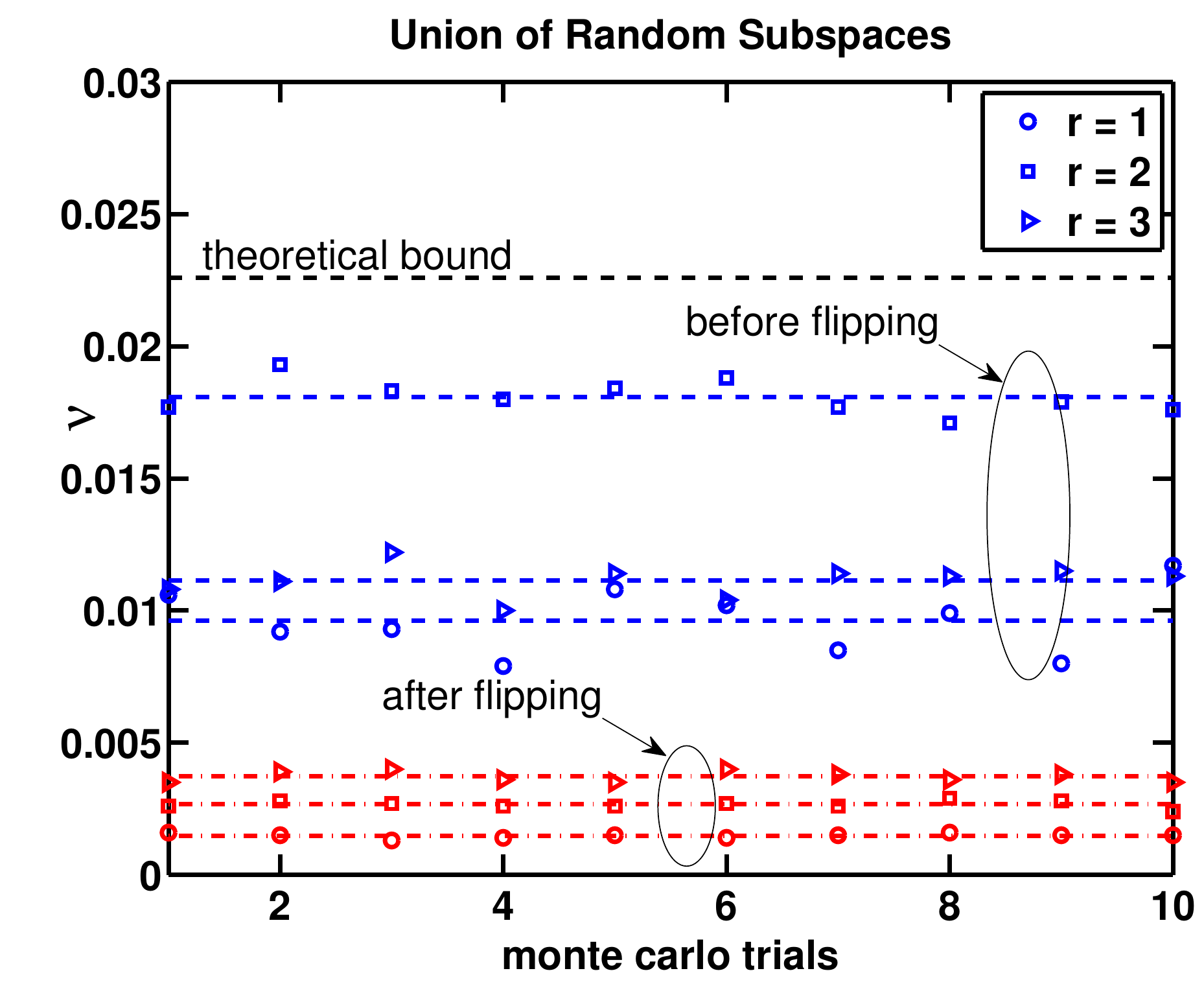}
\caption{Average block coherence of a random subspaces, for $r = 1, 2, 3$, respectively, before (blue) and after (red) passing through the flipping algorithm. The dashed lines represent mean of random realization of the average block coherence, before and after the flipping algorithm is applied. The theoretical bound corresponds to Lemma \ref{lemma2}.}
\label{fig:flipping}
\end{center}
\end{figure}

\section{Conclusion}

In this paper, we have studied worst-case and average block coherence of frames: from the perspective of compressed sensing, we are interested in matrices for which both of these quantities are low/optimal. We have shown that worst-case block coherence is subject to a fundamental lower bound, and we have presented deterministic constructions which achieve it. On the other hand, our analysis of random subspaces shows that they have only slightly sub-optimal worst-case block coherence. Meanwhile, low average block coherence can either be achieved by construction, or by means of a ``flipping'' algorithm designed to reduce average block coherence while preserving worst-case block coherence. A byproduct of our work is numerous constructions of families of optimal and almost optimal Grassmann packings.

Some questions remain: does there exist a fundamental lower bound on average block coherence, and if so which frames achieve it? What is the average block coherence in the case of random subspaces? Our optimal Grassmann packings give a smaller number of subspaces than the present theoretical limits: what is the true limit for the number of subspaces in an optimal Grassmann packing, and are there constructions which achieve it?

\appendix

\section{Proof of Theorem~\ref{random_result}}\label{random_proofs}

We first introduce some further definitions.

Given $p,q>0$ and $x\in[0,1]$, define the beta function~\cite[Sections 6.2]{handbook} to be
\begin{equation}\label{beta_fn}
B(p,q):=\int_0^1 t^{p-1}(1-t)^{q-1}\, dt,
\end{equation}
and define the regularized incomplete beta function~\cite[Section 6.6]{handbook} to be
\begin{equation}\label{reg_beta_fn}
I_x (p,q):=\frac{1}{B(p,q)}\int_0^x t^{p-1}(1-t)^{q-1}\, dt.
\end{equation}

The gamma and beta functions defined in (\ref{gamma_fn}) and (\ref{beta_fn}) are related as follows.

\begin{lemma}[{Beta-gamma identity~\cite[Theorem 7]{rainville}}]\label{beta_gamma}
For any $p,q>0$,
\begin{equation}\label{beta_gamma_identity}
B(p,q)=\frac{\Gm(p)\Gm(q)}{\Gm(p+q)}.
\end{equation}
\end{lemma}

We define $H(\rr)$, the Shannon entropy with base $e$ logarithms, in the usual way as

\begin{equation}\label{shannon_def}
H(\rr):=-\rr\ln\rr-(1-\rr)\ln(1-\rr),\;\;\;\;\mbox{for}\;\rr\in(0,1).
\end{equation}

The next result describes the asymptotic behavior of the gamma function.

\begin{lemma}[{Log-gamma asymptotic~\cite[Theorem 12 and (4)]{rainville}}]\label{binet}
For $p>0$,
$$\ln\Gm(p)=\left(p-\frac{1}{2}\right)\ln p-p+\frac{1}{2}\ln(2\pi)+o(1).$$
\end{lemma}

We may use Lemma~\ref{binet} to deduce the following limiting result for the beta function when the parameters grow proportionally.

\begin{lemma}[{Beta function limit}]\label{beta_limit}
Let $(p,q)\rightarrow\infty$ such that $p/(p+q)\rightarrow\rr$ as $(p,q)\rightarrow\infty$. Then
\begin{equation}\label{beta_lim}
\lim_{(p,q)\rightarrow\infty}\frac{1}{p+q}\ln B(p,q)=-H(\rr).
\end{equation}
\end{lemma}

\begin{proof}
By (\ref{beta_gamma_identity}) and (\ref{beta_lim}), we have
\begin{eqnarray}
&&\lim_{(p,q)\rightarrow\infty}\frac{1}{p+q}\ln B(p,q)\nonumber\\
&&=\lim_{(p,q)\rightarrow\infty}\frac{1}{p+q}\ln\left[\frac{\Gm(p)\Gm(q)}{\Gm(p+q)}\right]\nonumber
\end{eqnarray}
\begin{eqnarray}
&&=\lim_{(p,q)\rightarrow\infty}\frac{1}{p+q}[p\ln p-p+q\ln q-q-(p+q)\ln(p+q)+(p+q)]\nonumber\\
&&=\lim_{(p,q)\rightarrow\infty}\frac{1}{p+q}\ln\left[\frac{p^p q^q}{(p+q)^{p+q}}\right]\nonumber\\
&&=\lim_{(p,q)\rightarrow\infty}\frac{1}{p+q}\ln\left[\frac{\left(\frac{p}{p+q}\right)^p \left(\frac{q}{p+q}\right)^q}{\left(\frac{p+q}{p+q}\right)^{p+q}}\right]\nonumber\\
&&=\lim_{(p,q)\rightarrow\infty}\left[\frac{p}{p+q}\ln\left(\frac{p}{p+q}\right)+\frac{q}{p+q}\ln\left(\frac{q}{p+q}\right)\right]\nonumber\\
&&=\rr\ln\rr+(1-\rr)\ln(1-\rr),\nonumber
\end{eqnarray}
and the result follows by (\ref{shannon_def}).
\end{proof}

The next result gives the limiting behavior of the incomplete regularized beta function in the case where the parameters grow proportionally.

\begin{lemma}[{Regularized incomplete beta limit~\cite[Theorem 4.18]{mythesis}}]\label{reg_beta_lim}
Fix $0<x<1$ and let $p,q>0$ satisfy
$$\frac{p}{p+q}>x.$$
Let $(p,q)\rightarrow\infty$ such that
$$\lim_{p\ra\infty}\frac{p}{p+q}=\rr,$$
where $\rr>x$. Then
\begin{equation}\label{beta}
\lim_{p\ra\infty}\frac{1}{p+q}\ln I_x(p,q)=-\left[\rr\ln\left(\frac{\rr}{x}\right)+(1-\rr)\ln\left(\frac{1-\rr}{1-x}\right)\right].
\end{equation}
\end{lemma}

A random variable $X$ follows the (univariate) beta distribution $\mathrm{Beta}(a,b)$ for $a,b>0$ if it has probability density function (pdf)
$$f(x)=\frac{1}{B(a,b)}x^{a-1}(1-x)^{b-1};\;\;\;\;0\le x\le 1,$$
and we have the following result concerning its distribution function.

\begin{lemma}
Let $X\sim\mathrm{Beta}(a,b)$. Then its distribution function is
$$F(x)=I_x(a,b),$$
where $I_x(a,b)$ is the regularized incomplete beta function.
\end{lemma}

We may deduce from (\ref{beta_pdf}) the following bound on the pdf of the largest squared singular value $\lm_1$.

\begin{lemma}\label{pdf_bound}
Let $(\lm_1,\lm_2,\ldots,\lm_r)$ be distributed as in Lemma~\ref{random_pdf}. Then the following bound holds for the pdf of $\lm_1$.
\begin{equation}\label{bounded_pdf}
f(\lm_1)\le\frac{\sqrt{\pi}}{\left[B\left(\frac{r}{2}\right),\frac{n-r}{2}\right]^2}\cdot\lm_1^{\frac{1}{2}(2r-1)-1}(1-\lm_1)^{\frac{1}{2}(n-2r+1)-1},
\end{equation}
where $B(\cdot,\cdot)$ is defined in (\ref{beta_fn}).
\end{lemma}

\begin{proof}
Let $R$ be the region of $(r-1)$-dimensional space $\lm_i\geq 0;\;i=2,3,\ldots,r,$ and given $\lm_1>0$, let $R_{\lm_1}$ be the sub-region of $R$ consisting of all $(\lm_2,\ldots,\lm_r)$ such that $\lm_1\geq\lm_2\ldots\geq\lm_r\geq 0$. Then
\begin{eqnarray}
f(\lm_1)&=&c_{n,r}\lm_1^{-\frac{1}{2}}(1-\lm_1)^{\frac{1}{2}(n-2r-1)}\nonumber\\
&&\cdot\int_{R_{\lm_1}}\prod_{i<j}(\lm_1-\lm_j)\cdot\prod_{i=2}^r\lm_i^{-\frac{1}{2}}(1-\lm_i)^{\frac{1}{2}(n-2r-1)}\,d\lm_i\nonumber\\
&\le&c_{n,r}\lm_1^{r-\frac{3}{2}}(1-\lm_1)^{\frac{1}{2}(n-2r-1)}\nonumber\\
&&\cdot\int_{R_{\lm_1}}\prod_{2\le i<j}(\lm_1-\lm_j)\cdot\prod_{i=2}^r\lm_i^{-\frac{1}{2}}(1-\lm_i)^{\frac{1}{2}(n-2r-1)}\,d\lm_i\nonumber\\
&\le&c_{n,r}\lm_1^{r-\frac{3}{2}}(1-\lm_1)^{\frac{1}{2}(n-2r-1)}\nonumber\\
&&\cdot\int_{R}\prod_{2\le i<j}(\lm_1-\lm_j)\cdot\prod_{i=2}^r\lm_i^{-\frac{1}{2}}(1-\lm_i)^{\frac{1}{2}(n-2r-1)}\,d\lm_i\nonumber\\
&=&\frac{c_{n,r}}{c_{n-2,r-1}}\lm_1^{\frac{1}{2}(2r-1)-1}(1-\lm_1)^{\frac{1}{2}(n-2r+1)-1},\label{beta_bound}
\end{eqnarray}
where the last line follows by observing that we are integrating a multivariate beta distribution with $(n,r)\rightarrow(n-2,r-1)$. It remains to calculate $c_{n,r}/c_{n-2,r-1}$. We have
\begin{eqnarray}
\frac{c_{n,r}}{c_{n-2,r-1}}&=&\frac{\pi^{\frac{1}{2}r^2}\Gm_r\left(\frac{n}{2}\right)}{\left[\Gm_r\left(\frac{r}{2}\right)\right]^2\Gm_r\left(\frac{n-r}{2}\right)}\cdot\frac{\left[\Gm_{r-1}\left(\frac{r-1}{2}\right)\right]^2\Gm_{r-1}\left(\frac{n-r-1}{2}\right)}{\pi^{\frac{1}{2}(r-1)^2}\Gm_{r-1}\left(\frac{n-2}{2}\right)}\nonumber
\end{eqnarray}
\begin{eqnarray}
&=&\sqrt{\pi}\cdot\prod_{j=1}^r\frac{\Gm\left(\frac{n+1-j}{2}\right)}{\left[\Gm\left(\frac{r+1-j}{2}\right)\right]^2\Gm\left(\frac{n-r+1-j}{2}\right)}\cdot\prod_{j=1}^{r-1}\frac{\left[\Gm\left(\frac{r-j}{2}\right)\right]^2\Gm\left(\frac{n-r-j}{2}\right)}{\Gm\left(\frac{n-1-j}{2}\right)}\nonumber\\
&=&\sqrt{\pi}\cdot\frac{\Gm\left(\frac{n}{2}\right)\Gm\left(\frac{n-1}{2}\right)}{\left[\Gm\left(\frac{r}{2}\right)\right]^2\left[\Gm\left(\frac{n-r}{2}\right)\right]^2}\nonumber\\
&\le&\sqrt{\pi}\cdot\frac{\left[\Gm\left(\frac{n}{2}\right)\right]^2}{\left[\Gm\left(\frac{r}{2}\right)\right]^2\left[\Gm\left(\frac{n-r}{2}\right)\right]^2}\nonumber\\
&=&\frac{\sqrt{\pi}}{\left[B\left(\frac{r}{2},\frac{n-r}{2}\right)\right]^2},\nonumber
\end{eqnarray}
which combines with (\ref{beta_bound}) to prove (\ref{bounded_pdf}).
\end{proof}

Now let us write $\bar{F}(\lm_1)$ for complementary distribution function of $\lm_i$, which we next bound.

\begin{lemma}\label{beta_dist}
Let $(\lm_1,\lm_2,\ldots,\lm_r)$ be distributed as in Lemma~\ref{random_pdf}. Then the following bound holds for $\bar{F}(\lm_1)$.
\begin{eqnarray}\label{comp_bound}
\bar{F}(\lm_1)\le G(\lm_1)&:=&\sqrt{\pi}\cdot\frac{B\left[\frac{1}{2}(2r-1),\frac{1}{2}(n-2r+1)\right]}{\left[B\left(\frac{r}{2},\frac{n-r}{2}\right)\right]^2}\nonumber\\
&&\cdot I_{1-\lm_1}\left[\frac{1}{2}(n-2r+1),\frac{1}{2}(2r-1)\right],
\end{eqnarray}
where $B(\cdot,\cdot)$ is defined in (\ref{beta_fn}) and where $I_x(a,b)$ is defined in (\ref{reg_beta_fn}).
\end{lemma}

\begin{proof}
We may rewrite (\ref{pdf_bound}) as
\begin{equation}
f(\lm_1)\le\sqrt{\pi}\cdot\frac{B\left[\frac{1}{2}(2r-1),\frac{1}{2}(n-2r+1)\right]}{\left[B\left(\frac{r}{2},\frac{n-r}{2}\right)\right]^2}\cdot\frac{\lm_1^{\frac{1}{2}(2r-1)-1}(1-\lm_1)^{\frac{1}{2}(n-2r+1)-1}}{B\left[\frac{1}{2}(2r-1),\frac{1}{2}(n-2r+1)\right]},
\end{equation}
and we observe using Lemma~\ref{beta_dist} that the right-hand expression is the pdf of the $B\left[\frac{1}{2}(2r-1),\frac{1}{2}(n-2r+1)\right]$ distribution. We therefore have, by Lemma~\ref{beta_dist},
$$\begin{array}{rcl}
\bar{F}(\lm_1)&\le&\sqrt{\pi}\cdot\frac{B\left[\frac{1}{2}(2r-1),\frac{1}{2}(n-2r+1)\right]}{\left[B\left(\frac{r}{2},\frac{n-r}{2}\right)\right]^2}\cdot\left\{1-I_{\lm_1}\left[\frac{1}{2}(2r-1),\frac{1}{2}(n-2r+1)\right]\right\}\\
&=&\sqrt{\pi}\cdot\frac{B\left[\frac{1}{2}(2r-1),\frac{1}{2}(n-2r+1)\right]}{\left[B\left(\frac{r}{2},\frac{n-r}{2}\right)\right]^2}\cdot I_{1-\lm_1}\left[\frac{1}{2}(n-2r+1),\frac{1}{2}(2r-1)\right],
\end{array}$$
as required.
\end{proof}

The next lemma analyses the exponent of the upper tail of the distribution.

\begin{lemma}\label{exponent}
Let $r/n\rightarrow\beta\in(0,1/2)$ as $n\rightarrow\infty$. Then, for any $2\le a<1/\beta$,
\begin{equation}\label{G_limit}
\lim_{n\rightarrow\infty}\frac{1}{n}\ln G(a\beta)=\beta\ln a+\left(\frac{1-2\beta}{2}\right)\ln(1-a\beta)-(1-\beta)\ln(1-\beta),
\end{equation}
where $H(\cdot)$ is defined in (\ref{shannon_def}).
\end{lemma}

\begin{proof}
Applying (\ref{beta}), we obtain, for $a$ sufficiently large,
\begin{align*}
&\lim_{n\rightarrow\infty}\frac{1}{n}\ln I_{1-a\beta}\left[\frac{1}{2}(n-2r+1),\frac{1}{2}(2r-1)\right]\\
&=-\left[\left(\frac{1-2\beta}{2}\right)\ln\left(\frac{1-2\beta}{1-a\beta}\right)+\beta\ln\left(\frac{2\beta}{a\beta}\right)\right],
\end{align*}
which may be combined with (\ref{beta_lim}) to give
$$\lim_{n\rightarrow\infty}\frac{1}{n}\ln G(a\beta)=-\frac{1}{2}H(2\beta)+H(\beta)+\frac{1}{2}H(2\beta)+\beta\ln(a\beta)+\left(\frac{1-2\beta}{2}\right)\ln(1-a\beta),$$
which simplifies to give (\ref{G_limit}).
\end{proof}

This threshold will be used to upper bound the square of the spectral norm of two random subspaces, and hence the worst-case coherence. Note that the analysis presented here does not apply to the case where $\beta=1/2$, since in this case one of the parameters in the regularized incomplete beta function term is fixed and does not tend to infinity. We proceed to the proof of the main result.

\begin{proof}[Proof of Theorem~\ref{random_result}]
Consider a matrix $A_i^{\ast}A_j$ for some pair of blocks in $A$. It follows from (\ref{comp_bound}) and (\ref{G_limit}) that, for any $\eta>0$,
\begin{equation}\label{suff_large}
\frac{1}{n}\ln\PP\left\{\|A_i^{\ast}A_j\|_2^2\geq a\beta\right\}\le \beta\ln a+\left(\frac{1-2\beta}{2}\right)\ln(1-a\beta)-(1-\beta)\ln(1-\beta)+\eta,
\end{equation}
for all $n$ sufficiently large. By the definition of $\hat{a}(\beta)$ in (\ref{threshold_def}), and since the left-hand side of (\ref{threshold_def}) is strictly decreasing in $a$ for $\beta\in(0,1/2)$, for any $\e>0$, setting $a:=\hat{a}(\beta)+\e$, and choosing $\eta$ sufficiently small in (\ref{suff_large}) ensures
$$\frac{1}{n}\ln\PP\left\{\|A_i^{\ast}A_j\|_2^2\geq\hat{a}(\beta)\cdot\beta+\e\right\}\le -C\;\;\;\;\mbox{for all $n$ sufficiently large},$$
where $C$ is some positive constant, from which it follows that
\begin{equation}\label{single_pair}
\PP\left\{\|A_i^{\ast}A_j\|_2^2\geq\hat{a}(\beta)\cdot\beta+\e\right\}\le e^{-Cn}\;\;\;\;\mbox{for all $n$ sufficiently large}.
\end{equation}
Now we use (\ref{single_pair}) to union bound over all pairs of blocks, obtaining
$$\begin{array}{rcl}
\PP\Big\{[\mu(A)]^2\geq\hat{a}(\beta)\cdot\beta+\e\Big\}&=&\PP\Big\{\displaystyle\bigcup_{i\neq j}\left(\|A_i^{\ast}A_j\|_2^2\geq\hat{a}(\beta)\cdot\beta+\e\right)\Big\}\\
&\le&\displaystyle\sum_{i\neq j}\PP\left\{\|A_i^{\ast}A_j\|_2^2\geq\hat{a}(\beta)\cdot\beta+\e\right\}\\
&\le&\displaystyle\binom{m}{2}e^{-Cn}\ra 0\;\;\;\;\mbox{as}\;\;n\ra\infty,
\end{array}$$
since $m$ is polynomial in $n$.
\end{proof}

\bibliography{group_coherence_bib}

\end{document}